\documentclass{amsart}

\usepackage{amsmath}
\usepackage{amssymb}
\usepackage{amsthm}

\usepackage[foot]{amsaddr}

\begin{document}

\newtheorem{thrm}{Theorem}
\newtheorem{prop}[thrm]{Proposition}
\newtheorem{lem}[thrm]{Lemma}
\newtheorem{cor}[thrm]{Corollary}
\newtheorem{ex}[thrm]{Example}

\theoremstyle{remark}
\newtheorem*{rem}{Remark}

\title[A novel noncommutative KdV-type equation]{ A novel noncommutative KdV-type equation,
								 its recursion operator, and solitons. }										 
\maketitle

\quad

\centerline{\scshape Sandra Carillo$^{\rm {a,b}}$, Mauro Lo Schiavo$^{\rm {a}}$, 
				Egmont Porten$^{\rm {c,d}}$, Cornelia Schiebold$^{\rm {c,d}}$ }
\quad

{\footnotesize
 \centerline{$^{\rm a}$ Dipartimento SBAI, Sapienza Universit\`a di Roma, Rome, Italy }
 \centerline{$^{\rm b}$ I.N.F.N. - Sez. Roma1, Gr. IV - Mathematical Methods in NonLinear Physics, Rome, Italy}
 \centerline{$^{\rm c}$ DMA, Mid Sweden University, Sundsvall, Sweden}
 \centerline{$^{\rm d}$ Instytut Matematyki, Uniwersytet Jana Kochanowskiego w  Kielcach, Poland}
}

\quad

\begin{abstract}
    A noncommutative KdV-type equation is introduced extending the B\"acklund chart in \cite{Carillo:LoSchiavo:Schiebold:2016}. 
    This equation, called meta-mKdV here, is linked by Cole-Hopf transformations to the two noncommutative 
    versions of the mKdV equations listed in \cite[Theorem 3.6]{Olver:Sokolov:1998}.
    For this meta-mKdV, and its mirror counterpart, recursion operators, hierarchies and an explicit solution class are derived.
\end{abstract}

\quad

{\small
\noindent\emph{{\rm 1991} Mathematics Subject Classification.} 35Q53; 46L55; 37K35. \\[1ex]
}

\allowdisplaybreaks

\section{Introduction}

The study of integrable systems in \cite{Fuchssteiner:Carillo:1989} was based on a B\"acklund chart comprising the 
KdV and modified KdV equations, the KdV singularity manifold equation (also known as as UrKdV or Schwarz-KdV 
\cite{Weiss:1984,Wilson:1988}), the KdV interacting soliton equation \cite{Fuchssteiner:1987} and the Harry-Dym equation.
In the recent article \cite{Carillo:LoSchiavo:Schiebold:2016} a noncommutative interpretation of this chart is suggested. 
We mention two major features contrasting with the commutative model:

\quad\\[-4ex]

\indent 1.) At the place of the usual mKdV, we proceed via two different noncommutative  interpretations\footnote{
	Throughout the text,  $[A,B] = AB-BA$ and $\{ A, B \} = AB+BA$ denote commutator and anti-commutator
	of $A$ and $B$, respectively.}.
The first is the \emph{mKdV}
\begin{equation} \label{mkdv}
   V_t = V_{xxx} - 3 \{ V^2, V_x \} ,
\end{equation} 
see \cite{Athorne:Fordy:1987} for an early appearance. The second noncommutative version of the mKdV equation, 
which is abbreviated by \emph{alternative mKdV} in the sequel (\emph{amKdV} in \cite{Carillo:LoSchiavo:Schiebold:2016}), 
\begin{equation} \label{amkdv}
   \tilde V_t = \tilde V_{xxx} + 3 [ \tilde V, \tilde V_{xx} ] -6 \tilde V \tilde V_x \tilde V,
\end{equation} 
was first described by Khalilov and Khruslov \cite{Khalilov:Khruslov:1990}.

\quad\\[-4ex]

\indent 2.) The noncommutative B\"acklund chain in \cite{Carillo:LoSchiavo:Schiebold:2016}  does not include the Harry-Dym 
equation. The difficulty is that the interpretation of the reciprocal transformation between the KdV interacting soliton and the 
Harry-Dym equation\footnote{
	This transformation is an extended hodograph transformation, which intertwines dependent and independent variables.} 
requires substantially new ideas.

\quad\\[-4ex]

\indent
The relation between \eqref{mkdv} and \eqref{amkdv}, as presented in \cite{Carillo:LoSchiavo:Schiebold:2016}, 
builds on work by Liu and Athorne \cite{Athorne:Liu:1991}, who showed that the associated scattering problems 
are related by a gauge transformation. 
The fundamental observation of the present work is that, on the level of evolution equations, this correspondence 
can be made more explicit by introducing an intermediate equation
\begin{equation} \label{meta mkdv}
   Q_t = Q_{xxx} - 3 Q_{xx}Q^{-1}Q_{x},
\end{equation}
which we call \emph{meta-mKdV equation}, to emphasize that it is somewhat \emph{hidden} behind \eqref{mkdv} and \eqref{amkdv}. 
Analogously we obtain the mirror version\footnote{
	The term \emph{mirror} has been used in a comparable way for the noncommutative Burgers equations 
	in \cite{Kupershmidt:2005}, see also \cite{Carillo:LoSchiavo:Schiebold:arXiv2016,Carillo:Schiebold:JNMP2012}.
} 
of \eqref{meta mkdv} by reversing the order of multiplication in the nonlinear term.  
Those are linked to \eqref{mkdv}, \eqref{amkdv} by various Cole-Hopf transformations.
Including also the mirror meta-mKdV, we obtain the symmetric picture in Figure \ref{fig meta}.

The extended B\"acklund chart is then used in the further study of the novel equations. More precisely, recursion operators 
are derived, and it is explained why these operators are hereditary. Note that hereditariness is much harder to verify 
directly in the noncommutative setting \cite{Schiebold:2011}. This leads to hierarchies of commuting symmetries 
with the Cole-Hopf links extending to each level of the respective hierarchies. 

Moreover, we construct an explicit solution class of the meta-mKdV. Applying the Cole-Hopf link towards the mKdV, one recovers 
solutions which are already obtained in \cite{Carillo:Schiebold:2009} and can be interpreted as noncommutative analogs 
of 1-soliton solutions. In contrast, the solutions derived by applying the Cole-Hopf link towards the alternative mKdV are new to the 
best of the authors' knowledge. The fact that  the gauge transformation in \cite{Athorne:Liu:1991}  becomes completely explicit
for these solution classes is used to extend the construction to the entire hierarchies.

\quad

The article is organised as follows. In Section \ref{section meta} it is shown that two different noncommutative interpretations 
of the Cole-Hopf transformation map solutions of the meta-mKdV to solutions of the noncommutative mKdV and alternative mKdV,  
respectively. 
In Section \ref{section comparison} we make the comparison with the commutative case, where the picture is slightly different 
since the two noncommutative mKdV's specialise to one equation, the standard mKdV. 
Building on the present work, the commutative case was further elaborated in \cite{Carillo:2017}.              
It is interesting that the commutative meta-mKdV was also derived in the recent \cite{Horwath:Guengoer:2016}, 
from a completely different approach.
 
In Section \ref{section soliton} we give an explicit solution of the nc meta-mKdV, considered as an equation 
with values in some Banach algebra. As a corollary we obtain solutions of the nc mKdV's \eqref{mkdv} and \eqref{amkdv}, 
which can be understood as algebra-valued analogues of the familiar 1-soliton solutions. 
In Section \ref{section recursion operator} we use the Cole-Hopf link between \eqref{mkdv} and \eqref{amkdv} together 
with a structure theorem of Fokas and Fuchssteiner to derive a recursion operator for the nc meta-mKdV. 
In a way comparable to earlier work in \cite{Carillo:Schiebold:2009}, we then treat the induced hierarchy in 
Section \ref{section hierarchy} and extend the solutions from Section \ref{section soliton} to all members of the hierarchy. 

In Appendix \ref{section proof} we give the computationally involved proof of Theorem \ref{link from meta}.
Appendix \ref{section terminology} is a concise introduction to general methods for recursion operators and B\"acklund 
transformations. In Appendix \ref{section recursion operator mirror} the recursion operator of the nc mirror meta-mKdV 
is derived. Moreover, an alternative derivation of the recursion operators in Theorem \ref{rec op} is given.

\quad

\noindent{\bf Acknowledgements}

\smallskip

\noindent The financial support of G.N.F.M.-I.N.d.A.M., I.N.F.N. and {\sc Sapienza} University of Rome, Italy, are gratefully acknowledged. 
C. Schiebold wishes also to thank S.B.A.I. Dept. and {\sc Sapienza} University of Rome for the kind hospitality.

\quad

\section{A novel equation ``behind'' the two noncommutative mKdV equations} 
\label{section meta}

In this section we are concerned with the mathematical justification of Figure \ref{fig meta}. Throughout the article,
we understand evolution equations as \eqref{meta mkdv} as equations for an unknown function $Q(x,t)$ taking values
in a (possibly noncommutative) Banach algebra $\mathcal A$. If we consider B\"acklund transformations, like the two 
noncommutative interpretations of the Cole-Hopf transformation
\begin{eqnarray} 
   C(Q) &=& Q_x Q^{-1} ,		\label{cole hopf} \\
   \tilde C(Q) &=& Q^{-1} Q_x ,	\label{mirror cole hopf}
\end{eqnarray}
we will tacitly assume that the Banach algebra is the same for the involved equations.

The following is fundamental for the sequel.

\begin{thrm} \label{link from meta}
Let $Q$ be a solution of the meta-mKdV \eqref{meta mkdv}. Then
   \begin{enumerate}
      \item[a)] $V = C(Q)$ is a solution of the mKdV \eqref{mkdv}, 
      \item[b)] $\tilde V = \tilde C(Q)$ is a solution of the alternative mKdV \eqref{amkdv}.
   \end{enumerate}
\end{thrm}

The proof requires longer computations and is postponed to Appendix \ref{section proof}.

\quad

\begin{figure}

\unitlength0.9cm
\begin{picture}(15,6)

   \put(0,2.5){\framebox(3.6,1){\shortstack{\footnotesize mKdV \\ \footnotesize $V_t=V_{xxx} - 3\{V^2,V_x\}$}}}
   \put(10.5,2.5){\framebox(5,1){\shortstack{\footnotesize alternative mKdV \\
                         \footnotesize ${\tilde V}_t={\tilde V}_{xxx}+ 3[\tilde V,{\tilde V}_{xx}] - 6\tilde V{\tilde V}_x\tilde V$ }}}
                         
   \put(5.25,4.5){\framebox(4,1){\shortstack{\footnotesize meta-mKdV \\\footnotesize $Q_t = Q_{xxx} -3Q_{xx}Q^{-1}Q_{x}$}}}
   \put(5.25,0.5){\framebox(4,1){\shortstack{\footnotesize mirror meta-mKdV \\ 
   							\footnotesize $\tilde Q_t = \tilde Q_{xxx} -3\tilde Q_{x} \tilde Q^{-1}\tilde Q_{xx}$}}}

   \put(1.75,4.4){\footnotesize $V= Q_xQ^{-1}$}
  \put(1.75,1.25){\footnotesize $V= - \ \tilde Q^{-1}\tilde Q_x$}
   \put(7.5,3){\footnotesize $\tilde Q=Q^{-1}$}  
   \put(11,4.4){\footnotesize $\tilde V=Q^{-1}Q_x$}   
   \put(11,1.25){\footnotesize $\tilde V= - \  \tilde Q_x\tilde Q^{-1}$}

   \put(9.5,1){\vector(2,1){2.5}}   
   \put(9.5,5){\vector(2,-1){2.5}}   
   \put(5,1){\vector(-2,1){2.5}}   
   \put(5,5){\vector(-2,-1){2.5}}   

   \put(7.25,2){\vector(0,1){2.25}} \put(7.25,4){\vector(0,-1){2.25}}     

\end{picture}
\caption{The meta-mKdV and its mirror equation: connecting the two different noncommutative interpretations 
	of the mKdV equation.}
\label{fig meta}

\end{figure}
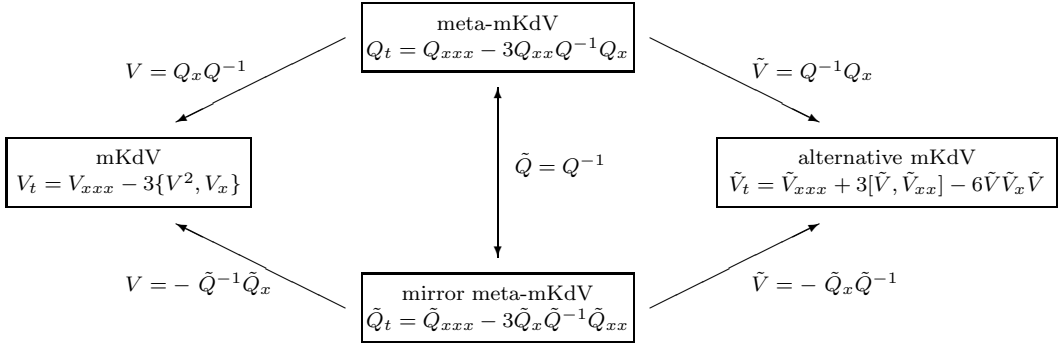

One gets a more complete picture by including the \emph{mirror meta-mKdV equation}, 
\begin{equation} \label{mirror meta mkdv}
    \tilde Q_t = \tilde Q_{xxx} - 3 \tilde Q_{x} \tilde Q^{-1} \tilde Q_{xx} ,
\end{equation}
in which the factors of the nonlinear term in \eqref{meta mkdv} appear in reversed order.

The next proposition establishes a direct link between the meta-mKdV \eqref{meta mkdv} and the mirror 
meta-mKdV \eqref{mirror meta mkdv}.

\begin{prop} \label{link between metas}
  	An invertible algebra-valued function $Q$ is a solution of \eqref{meta mkdv} if and only if $\tilde Q = Q^{-1}$ 
	is a solution of \eqref{mirror meta mkdv}.
\end{prop}

\begin{proof} Let $Q$ be a solution of the meta-mKdV \eqref{meta mkdv}, and $\tilde Q = Q^{-1}$. Then
 	\begin{eqnarray*}
	   \tilde Q_x &=& -Q^{-1}Q_xQ^{-1} , \\
	   \tilde Q_{xx} &=&  -Q^{-1}Q_{xx}Q^{-1} +  2Q^{-1}Q_xQ^{-1}Q_xQ^{-1} \\
	   \tilde Q_{xxx} &=& -Q^{-1}Q_{xxx}Q^{-1} + 3Q^{-1}Q_xQ^{-1}Q_{xx}Q^{-1} + 3Q^{-1}Q_{xx}Q^{-1}Q_xQ^{-1} \\
	   	&&   -6 Q^{-1}Q_xQ^{-1}Q_xQ^{-1}Q_xQ^{-1}  
	\end{eqnarray*}
	This shows $Q\big( \tilde Q_{xxx}  - 3 \tilde Q_{xx}  \tilde Q^{-1}  \tilde Q_{x} \big)Q = -( Q_{xxx} - 3Q_xQ^{-1}Q_{xx})$,
	 and similarly one verifies $Q\tilde Q_t Q = - Q_t$. As a result
	 \begin{equation*} 
	 	\tilde Q_{xxx}  - 3 \tilde Q_{xx}  \tilde Q^{-1}  \tilde Q_{x} 
	 		= -Q^{-1}( Q_{xxx} - 3Q_xQ^{-1}Q_{xx})Q^{-1} = -Q^{-1}Q_tQ^{-1} =\tilde Q_t .
	\end{equation*}
	showing that $\tilde Q$ solves \eqref{mirror meta mkdv}. The proof of the converse is symmetrical.
\end{proof}

Now we can prove the counterpart to Theorem \ref{link from meta} for the mirror meta-mKdV \eqref{mirror meta mkdv}.

\begin{prop} \label{link from mirror meta}
Let $\tilde Q$ be  a solution of the mirror meta-mKdV \eqref{mirror meta mkdv}. Then
\begin{enumerate}
      \item[a)] $V = - \ \tilde C(\tilde Q)$ is a solution of the mKdV \eqref{mkdv}, 
      \item[b)] $\tilde V = -\ C(\tilde Q)$ is a solution of the alternative mKdV \eqref{amkdv}.
   \end{enumerate}
\end{prop}

\begin{proof} Let $\tilde Q$ be a solution of \eqref{mirror meta mkdv}, and set $Q:=\tilde Q^{-1}$. Then   
   $\tilde Q_x = -  Q^{-1} Q_x Q^{-1}$, and therefore the relations
   \begin{eqnarray*}
      - C(\tilde Q) = -\tilde Q_x \tilde Q^{-1} = - (-Q^{-1} Q_x Q^{-1})Q = Q^{-1} Q_x = C(Q), \\
      - \tilde C(\tilde Q) = -  \tilde Q^{-1} \tilde Q_x = -  \tilde Q^{-1}  (-Q^{-1} Q_x Q^{-1}) = Q_x Q^{-1} = \tilde C(Q)
   \end{eqnarray*}
   hold. Furthermore, $Q$ is a solution of \eqref{meta mkdv} by Proposition \ref{link between metas}. 
   Now Proposition \ref{link from mirror meta} follows from Theorem \ref{link from meta}.
\end{proof}

\section{The commutative versus the noncommutative B\"acklund chart} 
\label{section comparison}

Including the results from Section \ref{section meta} into the noncommutative B\"acklund chart 
from \cite{Carillo:LoSchiavo:Schiebold:2016}, we arrive at the chart depicted in Figure \ref{fig nc chart}.
This chart starts with the KdV equation
\begin{equation} \label{kdv}
   U_t = U_{xxx} + 3 \{ U, U_x \} ,
\end{equation}
and proceeds via \eqref{mkdv}, \eqref{meta mkdv} and \eqref{amkdv} -- or \eqref{mkdv}, \eqref{mirror meta mkdv} 
and \eqref{amkdv} --  to the noncommutative counterparts of the KdV interacting soliton equation
\begin{equation} 
   \label{int so} S_t = S_{xxx} - \frac{3}{2} \big( S_xS^{-1}S_x \big)_x 
\end{equation}
and of the KdV singularity manifold equation 
\begin{equation} 
   \label{urkdv} \phi_t = \phi_{x} \{ \phi; x \} ,
\end{equation}
where $ \{ \phi; x \} $ denotes the noncommutative Schwarzian derivative
\begin{equation*}
   \{ \phi ; x \} = \big( \phi_x^{-1} \phi_{xx} \big)_x - \frac{1}{2} \big( \phi_x^{-1} \phi_{xx} \big)^2 .
\end{equation*}
The connecting B\"acklund links are given by
\begin{equation*}
   (a) \quad U = -\  (V^2+V_x), \qquad
   (b) \quad \tilde V = \frac{1}{2} S^{-1} S_x, \qquad
   (c) \quad S = \phi_x .
\end{equation*}
Note that $(a)$ is the noncommutative Miura transformation; the links $(d)$, $(\tilde d)$ stand for the Cole-Hopf transformation 
\eqref{cole hopf} and its mirror \eqref{mirror cole hopf}, compare also Theorem \ref{link from meta} and 
Proposition \ref{link from mirror meta}.

\begin{figure}[h]

\unitlength1cm
\begin{picture}(15,6.75)
   
   \put(0.5,5.5){\framebox(2,0.5){\shortstack{\footnotesize KdV$(U)$}}}
   \put(4,5.5){\framebox(2.5,0.5){\shortstack{\footnotesize mKdV$(V)$}}}
   
   \put(1.3,3){\framebox(2.5,1){\shortstack{\footnotesize \\ \footnotesize meta-mKdV$(Q)$}}}
    \put(6,3){\framebox(2.5,1){\shortstack{\footnotesize mirror \\ \footnotesize meta-mKdV$(\tilde Q)$}}}

   \put(4,0.5){\framebox(2.5,1){\shortstack{\footnotesize alternative \\ \footnotesize mKdV$(\tilde V)$}}}
   \put(8,0.75){\framebox(2.5,0.5){\shortstack{\footnotesize Int So KdV $(S)$}}}
   \put(12,0.75){\framebox(2.5,0.5){\shortstack{\footnotesize KdV Sing $(\phi)$}}}


   \put(0,0){\line(1,0){15}}
   \put(0,6.5){\line(1,0){15}}


   \put(3.75,5.75){\vector(-1,0){1}}   \put(3.1,6){\footnotesize $(a)$}
  
   \put(3.5,4.25){\vector(1,1){1}}		 \put(3.4,4.75){\footnotesize $(d)$}
   \put(3.5,2.75){\vector(1,-1){1}}	 \put(3.4,2){\footnotesize $(\tilde{d})$}
   \put(6.5,4.25){\vector(-1,1){1}} 	 \put(6.2,4.75){\footnotesize $-(\tilde{d})$}
   \put(6.5,2.75){\vector(-1,-1){1}}    	 \put(6.2,2){\footnotesize $-(d)$}
   
   \put(7.75,1){\vector(-1,0){1}}        \put(7.1,1.2){\footnotesize $(b)$}
   \put(11.75,1){\vector(-1,0){1}}      \put(11.1,1.2){\footnotesize $(c)$}

\end{picture}
\caption{KdV-type equations and their B\"acklund links: the non-com\-mu\-ta\-ti\-ve case.}
\label{fig nc chart}

\end{figure}

Since in the commutative case mKdV \eqref{mkdv} and alternative mKdV \eqref{amkdv} both reduce to the usual mKdV 
equation $v_t = v_{xxx} - 6 v^2v_x$, it is instructive to spell out how the meta-mKdV 
\begin{equation}
   \label{scalar meta mkdv} q_t = q_{xxx} - 3 \frac{q_xq_{xx}}{q} 
\end{equation}
fits into the B\"acklund chart in \cite{Fuchssteiner:Carillo:1989}.

By Theorem \ref{link from meta}, the Cole-Hopf transformation $v=q_x/q$ maps solutions of \eqref{scalar meta mkdv} 
to solutions of the mKdV. Moreover one sees that $s=q^2$ transfers solutions of \eqref{scalar meta mkdv} to solutions 
of the KdV interacting soliton equation (and vice versa as long as $\sqrt{s}$ is defined).
This means that we can split the B\"acklund link $(b)$ into two parts, namely
\begin{equation*}
   (b1) \quad v = \frac{q_x}{q}, \qquad
   (b2) \quad s = q^2 .
\end{equation*}
Observe that 
\begin{equation*}
	\frac{1}{2} \frac{s_x}{s} = \frac{1}{2} \frac{(q^2)_x}{q^2} = \frac{q_x}{q} = v .
\end{equation*}
Hence we get the extension of the B\"acklund chart in \cite{Fuchssteiner:Carillo:1989} as depicted in Figure \ref{fig scalar chart}.

\begin{figure}[h]

\unitlength1cm
\begin{picture}(15,1.5)
   
   \put(0,0.5){\framebox(1.5,0.5){\shortstack{\footnotesize KdV$(u)$}}}
   \put(2.5,0.5){\framebox(1.5,0.5){\shortstack{\footnotesize mKdV$(v)$}}}
   
   \put(5,0.5){\framebox(2.5,0.5){\shortstack{\footnotesize meta-mKdV$(q)$}}}
   
   \put(8.5,0.5){\framebox(2.5,0.5){\shortstack{\footnotesize Int So KdV $(s)$}}}
   \put(12,0.5){\framebox(2.5,0.5){\shortstack{\footnotesize KdV Sing $(\varphi)$}}}


   \put(-0.25,0){\line(1,0){15}}
   \put(-0.25,1.5){\line(1,0){15}}


   \put(2.4,0.7){\vector(-1,0){0.8}}        \put(1.9,0.9){\footnotesize (a)}
   \put(4.9,0.7){\vector(-1,0){0.8}}        \put(4.3,0.9){\footnotesize (b1)}
   \put(7.6,0.7){\vector(1,0){0.8}}        \put(7.7,0.9){\footnotesize (b2)}
   \put(11.9,0.7){\vector(-1,0){0.8}}       \put(11.4,0.9){\footnotesize (c)}

\end{picture}
\caption{KdV-type equations and their B\"acklund links: the commutative case.}
\label{fig scalar chart}

\end{figure}
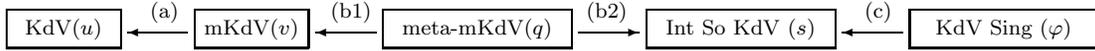

Building on the present work, the commutative case is further elaborated in \cite{Carillo:2017}.
For an occurrence of \eqref{scalar meta mkdv} in a completely different context, we refer to \cite{Horwath:Guengoer:2016}.

\section{Noncommutative analogs of solitons} 
\label{section soliton}

In the present section we obtain an explicit solution of the meta-mKdV \eqref{meta mkdv} depending on two 
parameters in $\mathcal A$, the Banach algebra where solutions take values. Using Theorem \ref{link from meta}, this solution 
is first transferred to the mKdV \eqref{mkdv}, giving new access to the soliton found in \cite{Carl:Schiebold:1999}, see also 
\cite{Carillo:Schiebold:2009} for the mKdV hierarchy. Secondly, we obtain the soliton for the alternative mKdV \eqref{amkdv}. 

\begin{prop} \label{soliton}
   Let $A, B\in \mathcal A$ with $A$ invertible, and let $L(x, t) = \exp(Ax + A^3t)B$. 
   Then a solution of the meta-mKdV equation \eqref{meta mkdv} is given by
   \begin{equation} \label{soliton meta}
      Q = (I - L)^{-1}A^{-1}(I + L) 
   \end{equation}
  on $\Omega=\{ (x,t) \in \mathbb{R}^2 \ | \ I \pm L(x,t)  \mbox{ \it are invertible} \}$.
\end{prop}

We will first show that  
\begin{equation} \label{soliton mirror meta}
   \tilde Q = (I + L)^{-1}A(I - L) ,
\end{equation}
which is the inverse of \eqref{soliton meta}, solves the mirror meta-mKdV \eqref{mirror meta mkdv}, and then use the link between
meta-mKdV and its mirror in Proposition \ref{link between metas}. To this end, we need two ingredients. 
First we observe

\begin{lem} \label{lem1}
   The function \eqref{soliton mirror meta} satisfies the equation $\tilde Q_{xx} = \tilde Q \tilde Q_x$.
\end{lem}

\begin{proof} Straightforward computations using the noncommutative product and derivation rules for inverses
give
\begin{eqnarray}
   \tilde Q_x &=& \big( (I + L)^{-1} \big)_x A (I - L) + (I + L)^{-1} A \big( I - L\big)_x \nonumber \\
      &=& - (I + L)^{-1} L_x (I + L)^{-1} \ A (I - L) + (I + L)^{-1} A \ ( - L_x) \nonumber \\
      &=& - (I + L)^{-1} AL (I + L)^{-1} A (I - L) + (I + L)^{-1} A ( - AL) \nonumber \\
      &=& - (I + L)^{-1} A (I + L)^{-1} \Big( LA (I - L) + (I + L) AL  \Big) \nonumber \\
      &=& - (I + L)^{-1} A (I + L)^{-1} (AL+LA) \label{eq1} ,
\end{eqnarray}
and then
\begin{eqnarray*}
   \tilde Q_{xx} &=& (I + L)^{-1}  AL  (I + L)^{-1} A  (I + L)^{-1} (AL+LA)   \\
   			&& \qquad  + (I + L)^{-1} A (I + L)^{-1}  AL  (I + L)^{-1} (AL+LA)  \\
			&& \qquad  - (I + L)^{-1} A (I + L)^{-1}  A (AL+LA)  \\
      &=& (I + L)^{-1} A (I + L)^{-1} \Big(LA+AL - A(I + L) \Big) (I + L)^{-1} (AL+LA)  \\
      &=& - (I + L)^{-1} A (I + L)^{-1}  \cdot (I - L)A \cdot (I + L)^{-1} (AL+LA)  \\
      &=& - (I + L)^{-1} A(I - L) \cdot (I + L)^{-1} A  (I + L)^{-1} (AL+LA)  \\
      &=& \tilde Q \tilde Q_x  ,
\end{eqnarray*}
which is the claim.
\end{proof}

The second ingredient is the close relation between \eqref{soliton mirror meta}  
and the soliton 
\begin{equation} \label{soliton pkdv}
      W = (I + L)^{-1}(AL + LA)
\end{equation}
of the potential KdV equation 
\begin{equation} \label{pkdv}
   W_t = W_{xxx} + 3 W_x^2
\end{equation}
derived in \cite[Corollary 9]{Carillo:Schiebold:2009}. More precisely, the following lemma holds.

\begin{lem} \label{lem2}
   For $\tilde Q$, $W$ as in \eqref{soliton mirror meta}, \eqref{soliton pkdv}, we have $\tilde Q = A - W$. 
\end{lem}

We are now in the position to prove  Proposition \ref{soliton}.

\begin{proof}[Proof of Proposition \ref{soliton}]
Note that Lemma \ref{lem2} implies that \eqref{soliton mirror meta} satisfies the potential KdV \eqref{pkdv} with a minus sign 
in front of the nonlinear term, $\tilde Q_t=\tilde Q_{xxx}-3 \tilde Q_x^2$. To prove that $\tilde Q$ solves the mirror 
meta-mKdV \eqref{mirror meta mkdv}, it hence remains to rewrite the nonlinear term $\tilde Q_x^2 
= \tilde Q_x \tilde Q^{-1}\tilde Q_{xx}$ using Lemma \ref{lem1}.
Finally,  since $Q=\tilde Q^{-1}$, Proposition \ref{soliton} follows from Proposition \ref{link between metas}.
\end{proof}

Using the Cole-Hopf transformations \eqref{cole hopf}, \eqref{mirror cole hopf} and Theorem \ref{link from meta},
it is straightforward to derive solutions of the mKdV \eqref{mkdv} and the alternative mKdV \eqref{amkdv} 
from the solution in Proposition \ref{soliton}.

\begin{cor} \label{cor soliton}
   Let $A, B \in \mathcal A$ with $A$ invertible, and let $L(x, t) = \exp(Ax + A^3t)B$. 
   Then 
   \begin{itemize}
      \item[a)] $V =  (I-L^2)^{-1} (AL+LA)$ is a solution of the mKdV \eqref{mkdv},
      \item[b)] $\tilde V =  (I + L)^{-1} A (I + L)^{-1} (AL+LA) (I - L)^{-1} A^{-1} (I + L) $ is a solution of the alternative 
      		mKdV equation \eqref{amkdv}.
   \end{itemize}
\end{cor}

\begin{proof} Let $\tilde Q$ be given as in \eqref{soliton mirror meta}.
   Starting with the expression for $\tilde Q_x$ derived in \eqref{eq1}, cf.~the proof of Lemma \ref{lem1}, 
   we have
   \begin{eqnarray*}
      \tilde Q_x &\stackrel{\eqref{eq1}}{=}& - (I + L)^{-1} A (I + L)^{-1} (AL+LA) \\
      	      &=& - (I + L)^{-1} A \bigm(  (I-L) \cdot (I-L)^{-1} \bigm) (I + L)^{-1}(AL+LA) \\
	      &=& -\tilde Q \cdot (I-L^2)^{-1} (AL+LA).
   \end{eqnarray*}
   On the other hand,
   \begin{eqnarray*}
      \tilde Q_x &=& - (I + L)^{-1} A (I + L)^{-1} (AL+LA) Q^{-1} \cdot Q \\
             &=& - (I + L)^{-1} A (I + L)^{-1} (AL+LA)    (I-L)^{-1} A^{-1} (I+L) \cdot  \tilde Q. 
   \end{eqnarray*}
   The proof is complete upon use of Proposition \ref{link from mirror meta}.
\end{proof}

\section{Recursion operators} 
\label{section recursion operator}

The aim of this section is to derive strong symmetries \cite{Fuchssteiner:1979} (or recursion operators in the sense
of \cite{Olver:1977}) of the meta-mKdV \eqref{meta mkdv} and its mirror \eqref{mirror meta mkdv}. 
A short summary of the related terminology is given in Appendix \ref{section terminology}. For more background information 
we refer to \cite{Fokas:Fuchssteiner:1981,Fuchssteiner:1979}, see also \cite{Schiebold:2011}.

Recall that the recursion operator 
\begin{eqnarray}
   \Psi (V) &=& (D - C_VD^{-1}C_V) (D - A_VD^{-1}A_V)  \label{mkdv rec} 
\end{eqnarray}
of the mKdV \eqref{mkdv} and the recursion operator  
\begin{eqnarray} 
   \tilde \Psi (\tilde V) &=& (D+2C_{\tilde V})(D-2R_{\tilde V}) (D+C_{\tilde V})^{-1} (D+2L_{\tilde V})D (D+C_{\tilde V})^{-1} 
   																		\nonumber \\ \label{amkdv rec}
      				&=& ({\tilde D}+C_{\tilde V}) (\tilde D-A_{\tilde V}) \tilde D^{-1} (\tilde D+A_{\tilde V})(\tilde D-C_{\tilde V}) \tilde D^{-1} 
\end{eqnarray}
of the alternative mKdV \eqref{amkdv} were derived in \cite{Guerses:Karasu:Sokolov:1999}, see 
also \cite{ Carillo:LoSchiavo:Schiebold:2016, Carillo:Schiebold:2009}.
In \eqref{amkdv rec} occurs the derivation\footnote{ 
That $\tilde D$ is a derivation means that $\tilde D (UV) = \tilde D(U)V+U\tilde D(V)$.}
\begin{equation} \label{derivation}
    \tilde D := D+C_{\tilde V} .
\end{equation}
As usual, $D$ stands for derivation with respect to $x$, $C_U$ and $A_U$ denote the commutator and anti-commutator 
with respect to $U$, and we will use $L_U$, $R_U$ for left- and right multiplication with $U$, respectively.

Using a structure theorem by Fokas and Fuchssteiner \cite{Fokas:Fuchssteiner:1981} about transferring strong 
symmetries along B\"acklund transformations, we deduce the meta-mKdV recursion operator \eqref{meta mkdv}   
from the Cole Hopf link in Theorem \ref{link from meta}.
Consider 
\begin{eqnarray*}
   B(Q,V) &=& VQ - Q_x ,
\end{eqnarray*}
which links the mKdV \eqref{mkdv} and the meta-mKdV \eqref{meta mkdv} according to Theorem \ref{link from meta}. 
Computing the Fr\'echet derivatives
\begin{eqnarray*}
	B_Q [\hat Q] &=& \frac{\partial}{\partial \epsilon} \Big|_{\epsilon=0}  \Big( V(Q+\epsilon \hat Q) -(Q+\epsilon \hat Q)_x \Big)  
					\ = \ V \hat Q - \hat Q_x \ = \ - (D-L_V) \ \hat Q ,\\
	B_V [\hat V] &=& \frac{\partial}{\partial \epsilon} \Big|_{\epsilon=0}  \Big( (V+\epsilon \hat V)Q - Q_x \Big)  
					\ = \ \hat V Q \ = \ R_Q \ \hat V ,
\end{eqnarray*}
we obtain the transformation operator
\begin{eqnarray*}
   \Pi &=& B_Q^{-1} B_V  \ = \ - (D-L_V)^{-1} R_Q .
\end{eqnarray*}
We will need another representation of $\Pi$.

\begin{lem} \label{identity}
   If $B(Q,V)=0$, then the identity $(D-L_V)^{-1} R_Q = R_Q (D-C_V)^{-1}$ holds.
\end{lem}
\begin{proof} 
   Using the product rule $DR_Q=R_QD+R_{Q_x}$ and the fact that left- and right multiplication commute, 
   we get
  \begin{equation*}
  	 (D-L_V)R_Q = DR_Q -L_VR_Q = R_QD+R_{Q_x} -R_QL_V = R_Q(D-L_V) + R_{Q_x}.
   \end{equation*}
   Since $B(Q,V)=0$, we have $Q_x=VQ$ and therefore $R_{Q_x}=R_{VQ}=R_QR_V$. This shows
   \begin{equation*}
      (D-L_V)R_Q = R_Q (D-L_V) + R_QR_V = R_Q (D-C_V),
   \end{equation*}
   which implies the lemma.
\end{proof}

By the structure theorem,
\begin{eqnarray*}
   \Phi(Q) &=& \Pi \ \Psi(V) \Pi^{-1} \\
      &=&  (D-L_V)^{-1} R_Q \ (D - C_VD^{-1}C_V) (D - A_VD^{-1}A_V) \  R_{Q^{-1}}(D-L_V) \\
      &=& R_Q (D-C_V)^{-1} \ (D - C_VD^{-1}C_V) (D - A_VD^{-1}A_V) \  (D-C_V) R_{Q^{-1}} \\
      &=& R_Q D^{-1} (D + C_V) (D - A_VD^{-1}A_V)  (D-C_V)  R_{Q^{-1}}  \\
      &=& R_Q D^{-1} (D + C_{Q_xQ^{-1}}) (D - A_{Q_xQ^{-1}}D^{-1}A_{Q_xQ^{-1}})  (D-C_{Q_xQ^{-1}})  R_{Q^{-1}} .
\end{eqnarray*}
is a recursion operator of the meta-mKdV \eqref{meta mkdv}. In the calculation above we have used Lemma \ref{identity}
and the factorisation $(D - C_VD^{-1}C_V) = (D-C_V)D^{-1}(D+C_V)$.

This implies part a) of the following theorem. For a verification of part b) see 
Appendix \ref{section recursion operator mirror}.

\begin{thrm}  \label{rec op} For the operators
   \begin{eqnarray*}
   	 \Phi(Q) &=& R_{Q} D^{-1} (D+C_{Q_x Q^{-1}}) (D-A_{Q_x Q^{-1}}D^{-1}A_{Q_x Q^{-1}}) (D-C_{Q_x Q^{-1}})R_{Q^{-1}}  , \\
	 \tilde \Phi(\tilde Q) &=& L_{\tilde Q} D^{-1} (D-C_{\tilde Q^{-1}\tilde Q_x})
	                                  			(D-A_{\tilde Q^{-1}\tilde Q_x}D^{-1}A_{\tilde Q^{-1}\tilde Q_x})
										   (D+C_{\tilde Q^{-1}\tilde Q_x})L_{\tilde Q^{-1}}  ,
   \end{eqnarray*}
   it holds that
   \begin{itemize}
   	\item[a)] $\Phi(Q)$ is a strong symmetry for the meta-mKdV \eqref{meta mkdv},
	\item[b)] $\tilde \Phi(\tilde Q)$ is a strong symmetry for the mirror meta-mKdV \eqref{mirror meta mkdv}.
   \end{itemize}
\end{thrm}

\quad

The next proposition explains the connection between $D$ and the derivation $\tilde D$ in \eqref{derivation}.
Denote by $K_Q$ conjugation with $Q$, i.e.
\begin{equation*} 
	K_Q = L_{Q^{-1}}R_Q .
\end{equation*}

\begin{prop}[{\cite[Proposition 3.3]{Carillo:LoSchiavo:Schiebold:2016}}] 
   For $V=Q_xQ^{-1}$ and $\tilde V = Q^{-1}Q_x$, we have
   \begin{itemize}
   	\item[a)]  $K_Q D K_Q^{-1}  = \tilde D$,
	\item[b)] $K_Q C_V K_{Q^{-1}} = C_{\tilde V}$ and $K_Q A_V K_{Q^{-1}} = A_{\tilde V}$.
   \end{itemize}
\end{prop}

In particular, $K_Q D^{-1} K_Q^{-1}  = \tilde D^{-1}$ and $K_Q (D+C_V) K_Q^{-1}  = \tilde D+C_{\tilde V}$,
implying 
\begin{equation*}
	K_Q D^{-1} (D+C_V) K_Q^{-1}  = K_Q D^{-1} K_{Q^{-1}} \cdot K_Q (D+C_V) K_Q^{-1} = \tilde D (\tilde D+C_{\tilde V}).
\end{equation*}
As a consequence, 
\begin{eqnarray*}
   \Phi(Q) &=& R_Q D^{-1} (D + C_V) (D - A_VD^{-1}A_V)  (D-C_V)  R_{Q^{-1}} \\
      &=& L_Q \Big( K_Q D^{-1} (D + C_V) (D - A_V) D^{-1} (D+A_V)  (D-C_V)  K_{Q^{-1}} \Big) L_{Q^{-1}} \\
      &=& L_Q \Big( {\tilde D}^{-1} ({\tilde D}+C_{\tilde V})
								({\tilde D}-A_{\tilde V}) {\tilde D}^{-1}(\tilde D +A_{\tilde V}) 
									 ({\tilde D}-C_{\tilde V}) \Big) L_{Q^{-1}} \\
      &=& L_Q {\tilde D}^{-1} ({\tilde D}+C_{\tilde V})
								({\tilde D}-A_{\tilde V}{\tilde D}^{-1}A_{\tilde V}) 
									 ({\tilde D}-C_{\tilde V}) L_{Q^{-1}}.
\end{eqnarray*}
Hence conjugation gives the following alternative representation of the operators in Theorem \ref{rec op}.

\begin{prop} \label{rec op derivation}
   Define the derivations $\mathbb{D} = D + C_{Q^{-1}Q_x}$ and $\mathbb{\tilde D} = D - C_{\tilde Q_x\tilde Q^{-1}}$.
   Then the operators in Theorem \ref{rec op} can be rewritten as
   \begin{eqnarray*}
		   \Phi(Q) &=& L_{Q} \mathbb{D}^{-1} (\mathbb{D}+C_{Q^{-1}Q_x})
								(\mathbb{D}-A_{Q^{-1}Q_x}\mathbb{D}^{-1}A_{Q^{-1} Q_x}) 
									 (\mathbb{D}-C_{Q^{-1} Q_x}) L_{Q^{-1}} ,
			\\
		   \tilde \Phi(\tilde Q) &=& R_{\tilde Q} \mathbb{\tilde D}^{-1} (\mathbb{\tilde D}-C_{\tilde Q_x\tilde Q^{-1}})
								(\mathbb{\tilde D}-A_{\tilde Q_x\tilde Q^{-1}}\mathbb{\tilde D}^{-1}A_{\tilde Q_x\tilde Q^{-1}}) 
											  (\mathbb{\tilde D}+C_{\tilde Q_x\tilde Q^{-1}}) R_{\tilde Q^{-1}} .
   \end{eqnarray*}
\end{prop}

\begin{rem}
   It is worth mentioning that the strong symmetries one obtains transferring the recursion operator \eqref{amkdv rec} 
   to the meta-mKdV \eqref{meta mkdv} and the mirror meta-mKdV \eqref{mirror meta mkdv} are precisely those 
   given in the above proposition, see Appendix \ref{section recursion operator mirror}.
\end{rem}

Finally, we turn to hereditariness, an important concept introduced in \cite{Fuchssteiner:1979}. Hereditariness of the
noncommutative KdV recursion operator 
\begin{equation*}
   \Phi_{\rm KdV} (U) = D^2 + 2 A_U + A_{U_x} D^{-1} + C_U D^{-1} C_U D^{-1}
\end{equation*}
is the main result of \cite{Schiebold:2011}. Another result from \cite{Fokas:Fuchssteiner:1981} ensures that 
hereditariness is preserved by B\"acklund transformations. Applying first the Miura transformation and 
then the respective Cole-Hopf link yields

\begin{thrm}
   The operators in Theorem \ref{rec op} are hereditary symmetries.
\end{thrm}

\section{Solving the meta-mKdV hierarchy} 
\label{section hierarchy}

Consider the \emph{noncommutative meta mKdV hierarchy},
\begin{equation*} 
   (E_{2n-1}) \hspace*{2cm} Q_{t_{2n-1}} =  \Phi(Q)^{n-1} Q_x ,
\end{equation*}
for $n\geq 1$, where $\Phi(Q)$ is the recursion operator given in Theorem \ref{rec op}. The base member ($n=1$) 
is $Q_{t_1} = Q_x$, and the following argument shows that for $n=2$ the meta-mKdV \eqref{meta mkdv} is obtained:

Abbreviate $V=Q_xQ^{-1}$. As a first step, $(D-C_V) R_{Q^{-1}} Q_x = (D-C_V) V = V_x$. Second, since $D^{-1} A_V V_x 
= D^{-1} \{ V, V_x \} = V^2$, we find $(D-A_{V}D^{-1}A_{V}) V_x = V_{xx} - 2 V^3$. This shows
\begin{eqnarray*}
 \Phi(Q) Q_x &=& R_{Q} D^{-1} (D+C_{V}) (D-A_{V}D^{-1}A_{V}) (D-C_{V})R_{Q^{-1}} Q_x \\
  &=& R_{Q} D^{-1} (D+C_{V}) \Big( V_{xx} - 2  V^3 \Big) \\
  &=& R_{Q} D^{-1} \Big( (V_{xx} - 2  V^3)_x - \big[ V, V_{xx} \big] \Big) \\ 
  &=&  \Big( (V_{xx} - 2  V^3) - \big[ V, V_{x} \big]  \Big) Q .
\end{eqnarray*}
Reinserting $V=Q_xQ^{-1}$ and computing the term in the large brackets explicitly, one obtaines $Q_{xxx}-3Q_{xx} Q^{-1} Q_x$, 
i.e. the right-hand side of \eqref{meta mkdv}.
\begin{rem}
   Observe that this gives a heuristic access to the meta-mKdV \eqref{meta mkdv}.
\end{rem}

The equations $(E_{2n-1})$, $1\leq n\leq N$, are regarded as a system of partial differential equations in the variables 
$t_1, t_3 , \ldots , t_{2N-1}$ (note that we as usual identify $t_1=x$, $t_3=t$). In the literature, it is customary to consider 
$\{ (E_{2n-1}) \}_{n\geq 1}$ as an infinite system for formal functions in infinitely many variables $t_1,t_3,\ldots$. 
In our context we prefer to work with truncated expressions in finitely many variables.

The following theorem generalises Proposition \ref{soliton}.

\begin{thrm} \label{soliton hierarchy}
   Let $A, B\in \mathcal A$ with $A$ invertible. Consider
   \begin{equation*}
      L_N (t_1, \ldots , t_{2N-1}) = \exp \big( \sum_{k=1}^N A^{2k-1}t_{2k-1} \big) B .
    \end{equation*}
   Then, for all $N \in \mathbb{N}$, a solution of the system of the first $N$ equations of the meta-mKdV hierarchy 
   is given by
   \begin{equation*}
      Q_N = (I - L_N)^{-1}A^{-1}(I + L_N)
   \end{equation*}
   on $\Omega = \{ (x,t) \ | \ I \pm L_N \mbox{ \it are invertible} \}$.
\end{thrm}

The proof heavily relies on \cite[Theorem 12]{Carillo:Schiebold:2009}, where it is shown that
\begin{equation*}
      V_N = (I-L_N^2)^{-1} (AL_N+L_NA)
\end{equation*} 
solves the system of the first $N$ equations of the noncommutative mKdV hierarchy $V_{t_{2n-1}} =\Psi(V)^{n-1} V_x$,
$1\leq n\leq N$, where $\Psi(V)$ is the recursion operator in \eqref{mkdv rec}. Then the Cole-Hopf 
transformation \eqref{cole hopf} is used to transfer this solution to the meta-mKdV 
hierarchy. 

\begin{proof}[Proof of Theorem \ref{soliton hierarchy}] 
      To simplify notation, we suppress the index $N$ in the rest of the proof. 
      Set $V=C(Q)$, where $C$ is given in \eqref{cole hopf}.
      In the proof of Corollary \ref{cor soliton}, it has
      been verified that
      \begin{eqnarray*}
         V = (I-L^2)^{-1} (AL+LA) .
      \end{eqnarray*} 
      By \cite[Theorem 12]{Carillo:Schiebold:2009}, $V$ is a solution of the mKdV hierarchy.
      Furthermore,  
      \begin{eqnarray*}
              V_{t_{2n-1}} &=& \big( Q_xQ^{-1} \big)_{t_{2n-1}} 
         			\ =\  Q_{xt_{2n-1}} Q^{-1} - Q_x Q^{-1} Q_{t_{2n-1}} Q^{-1}  
			 	\ =\  R_{Q^{-1}}(D-L_V) Q_{t_{2n-1}} \\
			       &=& \Pi^{-1} \ Q_{t_{2n-1}}  ,
      \end{eqnarray*}   
      where $\Pi =(D-L_V)^{-1}R_Q$.
      Hence,
      \begin{eqnarray*}
      		 \Pi^{-1} \ Q_{t_{2n-1}} =V_{t_{2n-1}}  = \Psi (V)^{n-1}  V_x  = \Psi (V)^{n-1} \  \Pi^{-1} \ Q_x .
      \end{eqnarray*}
      Recall from the construction of the meta-mKdV recursion operator preceding Theorem \ref{rec op}
      that $\Phi(Q) = \Pi \ \Psi(V) \ \Pi^{-1}$. As a result,
       \begin{eqnarray*}
          Q_{t_{2n-1}} = \Pi \bigm( \Pi^{-1}\ \Phi(Q) \ \Pi \bigm)^{n-1}\  \Pi^{-1}\ Q_x  = \Phi(Q)^{n-1} Q_x ,
      \end{eqnarray*}
      which completes the proof.
\end{proof}

Finally, the solution given in Theorem \ref{soliton hierarchy} is mapped\footnote{
Observe first that the Cole-Hopf transformations \eqref{cole hopf}, \eqref{mirror cole hopf} also map solutions 
of the base member of the meta-mKdV hierarchy, $Q_t =Q_x$, to solutions of the base members $V_t=V_x$ 
and $\tilde V_t = \tilde V_x$ of the mKdV and the alternative mKdV hierarchies. Since the recursion operators 
are hereditary, a result in \cite{Fokas:Fuchssteiner:1981} then ensures that the Cole-Hopf links extend to all levels 
of the hierarchies. 

Note that this also yields an alternative (less direct) proof of Theorem \ref{link from meta}. 
}
to a solution of the alternative mKdV hierarchy, 
\begin{equation*}
   \tilde V_{t_{2n-1}} = \tilde \Psi (\tilde V)^{n-1} \tilde V_x,
\end{equation*}
$1\leq n\leq N$, where $\tilde \Psi(\tilde V)$ is the recursion operator in \eqref{amkdv rec}. This leads to

\begin{cor}  Let the assumptions of Theorem \ref{soliton hierarchy} be satisfied.    
   Then, for all $N \in \mathbb{N}$, a solution of the system of the first $N$ equations of the alternative mKdV hierarchy 
   is given by
   \begin{equation*}
       \tilde V_N = (I + L_N)^{-1} A (I + L_N)^{-1} (AL_N+L_NA) (I - L_N)^{-1} A^{-1} (I + L_N)  
   \end{equation*}
   on $\Omega = \{ (x,t) \ | \ I \pm L_N \mbox{ \it are invertible} \}$.
\end{cor}

\begin{appendix}

\section{Proof of Theorem \ref{link from meta}} 
\label{section proof}

In this appendix we give a direct verification that the Cole-Hopf transformations \eqref{cole hopf} and \eqref{mirror cole hopf}
map solutions of the meta-mKdV \eqref{meta mkdv} to solutions of the mKdV \eqref{mkdv} and the alternative 
mKdV \eqref{amkdv}, respectively.

\begin{proof}[Proof of Theorem \ref{link from meta}] Let $Q$ be a solution of the meta-mKdV \eqref{meta mkdv}. 
   As for (a), we need to verify that  $V=C(Q)=Q_xQ^{-1}$ solves the mKdV \eqref{mkdv}. We 
   compute\footnote{In the subsequent computations only the noncommutative product rule and derivation
   rule for inverses, $(Q^{-1})_x = Q^{-1}Q_xQ^{-1}$, are used.}
   \begin{eqnarray*}
	V_x &=& \big( Q_xQ^{-1}\big)_x \ = \ Q_{xx}Q^{-1} - (Q_xQ^{-1})^2 , \\
      	V_{xx} &=& \big( Q_{xx}Q^{-1} - (Q_xQ^{-1})^2 \big)_x \\
		&=& Q_{xxx}Q^{-1} - 2Q_{xx}Q^{-1}Q_xQ^{-1}
      					- Q_xQ^{-1}Q_{xx}Q^{-1} + 2(Q_xQ^{-1})^3 , \\
      	V_{xxx} &=&  Q_{xxxx}Q^{-1} \\
			 && \qquad - 3 Q_{xxx} Q^{-1} Q_x Q^{-1}
					- 3 Q_{xx} Q^{-1} Q_{xx} Q^{-1}
					-  Q_x Q^{-1} Q_{xxx} Q^{-1} \\
      			&& \qquad	+ 6 Q_{xx} Q^{-1} Q_x Q^{-1} Q_x Q^{-1}
					+ 3 Q_x Q^{-1} Q_{xx} Q^{-1} Q_x Q^{-1}
					+ 3 Q_x Q^{-1} Q_x Q^{-1} Q_{xx}  Q^{-1} \\			
			&& \qquad - 6( Q_x Q^{-1})^4 .
   \end{eqnarray*}
   For the nonlinear term in \eqref{mkdv}  we get
   \begin{eqnarray*}
      \{ V^2,V_x\} &=&   Q_x Q^{-1} Q_x Q^{-1} Q_{xx} Q^{-1}
      					+  Q_{xx} Q^{-1} Q_x Q^{-1} Q_x Q^{-1} - 2( Q_x Q^{-1})^4 ,
   \end{eqnarray*}
   and hence
   \begin{eqnarray*}
      V_{xxx}- 3 \{ V^2,V_x\} 
      	&=&  Q_{xxxx} Q^{-1} \\
		&&  \qquad - 3 Q_{xxx} Q^{-1} Q_x Q^{-1}
					- 3 Q_{xx} Q^{-1} Q_{xx} Q^{-1}
					-  Q_x Q^{-1} Q_{xxx} Q^{-1} \\
      		&&	\qquad+ 3 Q_{xx} Q^{-1} Q_x Q^{-1} Q_x Q^{-1}
					+ 3 Q_x Q^{-1} Q_{xx} Q^{-1} Q_x Q^{-1} .
   \end{eqnarray*}
   On the other hand, we find
   \begin{eqnarray*}
      V_t &=& Q_{xt}Q^{-1} - Q_x Q^{-1}Q_tQ^{-1} \\
      	&=& \big( Q_{xxx}-3Q_{xx}Q^{-1}Q_{x} \big)_x Q^{-1} - Q_xQ^{-1} \big( Q_{xxx}-3Q_{xx}Q^{-1}Q_{x} \big) Q^{-1} ,
   \end{eqnarray*}
   since $Q$ solves \eqref{meta mkdv}, and, consequently,
   \begin{eqnarray*}
	V_t &=& \big( Q_{xxxx} - 3 Q_{xxx}Q^{-1}Q_{x} - 3 Q_{xx}Q^{-1}Q_{xx} + 3 Q_{xx}Q^{-1}Q_xQ^{-1}Q_{x} \big) Q^{-1}  \\
		&&\qquad  -Q_xQ^{-1} \big( Q_{xxx} - 3 Q_{xx}Q^{-1}Q_{x} \big) Q^{-1} \\
	&=&  Q_{xxxx}Q^{-1} - 3 Q_{xxx}Q^{-1}Q_{x}Q^{-1}  - 3 Q_{xx}Q^{-1}Q_{xx}Q^{-1} - Q_xQ^{-1}Q_{xxx}Q^{-1} \\
		&&\qquad  + 3Q_{xx}Q^{-1}Q_xQ^{-1}Q_{x} Q^{-1} + 3Q_xQ^{-1}Q_{xx}Q^{-1}Q_{x} Q^{-1} \\
	&=& V_{xxx}- 3 \{ V^2,V_x\} .
   \end{eqnarray*}
   
   As for (b), one has to check that $\tilde V=\tilde C(Q) = Q^{-1}Q_x$ solves the alternative mKdV \eqref{amkdv}.
   Here we compute
   \begin{eqnarray*}
      \tilde V_x &=& \big(Q^{-1}Q_x\big)_x \ = \ -(Q^{-1}Q_x)^2+Q^{-1}Q_{xx} , \\
      \tilde V_{xx} &=&  2(Q^{-1}Q_x)^3-Q^{-1}Q_{xx}Q^{-1}Q_x-2Q^{-1}Q_xQ^{-1}Q_{xx}+Q^{-1}Q_{xxx} , \\
      \tilde V_{xxx} &=& -6(Q^{-1}Q_x)^4 \\
      			&& \qquad  +3Q^{-1}Q_{xx}Q^{-1}Q_xQ^{-1}Q_x+3Q^{-1}Q_xQ^{-1}Q_{xx}Q^{-1}Q_x+6Q^{-1}Q_xQ^{-1}Q_xQ^{-1}Q_{xx} \\
			&& \qquad -Q^{-1}Q_{xxx}Q^{-1}Q_x-3Q^{-1}Q_{xx}Q^{-1}Q_{xx}-3Q^{-1}Q_xQ^{-1}Q_{xxx} \\
			&& \qquad +Q^{-1}Q_{xxxx} .
   \end{eqnarray*}
   For the nonlinear terms in \eqref{amkdv} we therefore get
    \begin{eqnarray*}
      	[\tilde V,\tilde V_{xx}] &=& Q^{-1} Q_x Q^{-1} Q_{xxx} - Q^{-1} Q_{xxx}Q^{-1} Q_x \\
						&& \quad + Q^{-1} Q_{xx} Q^{-1} Q_xQ^{-1} Q_x
							+ Q^{-1} Q_x Q^{-1} Q_{xx} Q^{-1} Q_x  - 2 Q^{-1} Q_x Q^{-1} Q_x Q^{-1} Q_{xx}  , \\
	-2 \tilde V\tilde V_x\tilde V &=& 2 ( Q^{-1} Q_x)^4 -2   Q^{-1} Q_x Q^{-1} Q_{xx}  Q^{-1} Q_x .
   \end{eqnarray*}
  This shows
   \begin{eqnarray*}
      \tilde V_{xxx} +3 [\tilde V,\tilde V_{xx}]  -6 \tilde V\tilde V_x\tilde V
      			&=&	6 Q^{-1} Q_{xx} Q^{-1} Q_x Q^{-1} Q_x \\
			&& \qquad  -4 Q^{-1} Q_{xxx} Q^{-1} Q_x-3 Q^{-1} Q_{xx} Q^{-1} Q_{xx} 
			 											 + Q^{-1} Q_{xxxx}  .
   \end{eqnarray*}
   On the other hand,
   \begin{eqnarray*}
      \tilde V_t &=&  \big( Q^{-1} Q_x\big)_t \ = \  - Q^{-1} Q_t Q^{-1} Q_x + Q^{-1} Q_{xt} \\
      	&=& -  Q^{-1} \big(  Q_{xxx}-3 Q_{xx} Q^{-1} Q_x \big)  Q^{-1} Q_x
			+ Q^{-1} \big( Q_{xxx}-3 Q_{xx} Q^{-1} Q_x \big)_x 
   \end{eqnarray*}
   since $Q$ solves \eqref{mkdv}. Hence
   \begin{eqnarray*}
	\tilde V_t &=&  Q^{-1}  Q_{xxxx} - 4  Q^{-1}  Q_{xxx}  Q^{-1} Q_x \\
		&&  \qquad  - 3  Q^{-1}  Q_{xx} Q^{-1} Q_{xx} 
			+ 6  Q^{-1}  Q_{xx} Q^{-1} Q_x Q^{-1} Q_x \\
		&=& \tilde V_{xxx} +3 [\tilde V,\tilde V_{xx}]  -6 \tilde V\tilde V_x\tilde V .
    \end{eqnarray*}
    This completes the proof.
\end{proof}

\section{Background on symmetries and B\"acklund transformations} 
\label{section terminology}

In this appendix, we will concisely review some general concepts for evolution equations, 
which are used throughout the article.
The essence is to view an evolution equation as an ordinary differential equation with values 
in a suitable function space and  to apply methods from differential topology and dynamical systems.
We refer to \cite{Marsden:Ratiu:1999} for an introduction to the general theory, and 
to \cite{Schiebold:2011} and the references therein for more details on applications relevant 
for the present article.

Our phase space will be a topological algebra $\mathcal{F}$ of functions depending on $x\in\mathbb{R}$ 
and with values in some Banach algebra $\mathcal{A}$. 
An evolution equation can be viewed as an ordinary differential equation
\begin{equation}\label{vector field K}
U_t=K(U),
\end{equation}
where the unknown function $U=U(t)$ takes values in $\mathcal{F}$ and $K$ is a vector field 
on $\mathcal{F}$, i.e.~a mapping that associates to $U\in\mathcal{F}$ a vector $K(U)$ in the 
tangent space $T_U \mathcal{F}$ of $\mathcal{F}$ at $U$. 
In the sequel, all vector fields are assumed to be smooth.
The Lie bracket $[K,L]_{\rm Lie}$ of two vector fields should not be confused with their commutator 
$[K,L]$, defined pointwise by $[K,L](U)=[K(U),L(U)]$, where we identify $T_U \mathcal{F}$ with 
$\mathcal{F}$ and use the commutator in $\mathcal{A}$.

A vector field $L$ is a \emph{symmetry of \eqref{vector field K}} if $[K,L]_{\rm Lie}\equiv 0$.
By an operator we mean a mapping $\Phi$ that maps $U\in\mathcal{F}$ to a bounded linear operator 
$\Phi(U):T_U \mathcal{F}\rightarrow T_U \mathcal{F}$
on the tangent space at $U$. 
Such a $\Phi$ is a \emph{strong symmetry  of \eqref{vector field K}} (or recursion operator) 
if $\Phi$ maps symmetries to symmetries.
As proved in \cite{Fuchssteiner:1979}, $\Phi$ is a strong symmetry  of \eqref{vector field K} if
\begin{equation*}
	\Phi'[K] V=K'[\Phi V]-\Phi K'[V]
\end{equation*}
is valid for every vector field $V$.
An operator $\Phi$ is \emph{hereditary} if 
\begin{equation*}
	\Phi\Phi'[V]W-\Phi'[\Phi V]W
\end{equation*}
is symmetric in the vector fields $V$ and $W$.
Note that hereditariness does not depend on the underlying evolution equation.

If $\Phi$ is a strong symmetry, so also the powers $\Phi^n$, $n\in\mathbb{N}$,
and we get a hierarchy of vector fields $K_n=\Phi^n K$ which all commute with $K$.
For $\Phi$ hereditary, it was proved in \cite{Fuchssteiner:1991} that these vector fields \emph{all}
commute pairwise. 
In all cases considered in the present article, the hierarchy can be extended taking as a base member 
$K_0(U)=U_x$, the symmetry expressing translation invariance.

Let $\mathcal{F}$, $\mathcal{G}$ and $\mathcal{H}$ be
function spaces as above. 
In addition to \eqref{vector field K}, consider a second evolution equation
\begin{equation}\label{vector field L}
V_t=L(V),
\end{equation}
for $V=V(t)$ with values in $\mathcal{G}$.
Following \cite{Fokas:Fuchssteiner:1981}, a smooth mapping
\begin{equation*}
B:\mathcal{F}\times \mathcal{G}\rightarrow \mathcal{H}
\end{equation*}
defines a \emph{B\"acklund transformation} if 
\begin{itemize}
   \item[1)] for any solutions $U(t)$, $V(t)$ of \eqref{vector field K}, \eqref{vector field L}, respectively,
		which are defined on an open interval $I$ and satisfy $B(U(t_0),V(t_0))=0$ for some $t_0\in I$, 
		we have $B(U(t),V(t))=0$ for all $t\in I$.
\end{itemize}
Such a B\"acklund transformation is called \emph{locally invertible}, if
\begin{itemize}
   \item[2)] at every point $(U,V)\in \mathcal{F}\times \mathcal{G}$, the partial derivatives $B_U$ and $B_V$ 
		restrict to invertible linear mappings $T_{U}\mathcal{F}\times\{0\}\rightarrow T_{B(U,V)}\mathcal{H}$ 
		and $\{0\}\times T_{V}\mathcal{G}\rightarrow T_{B(U,V)}\mathcal{H}$, respectively. 
		Here we use the obvious identification $T_{(U,V)}\cong T_{U}\mathcal{F}\times T_{V}\mathcal{G}
		\cong\mathcal{F}\times\mathcal{G}$,
\end{itemize}
Let $\Phi$ be a strong symmetry of \eqref{vector field K} and let $B$ define a locally invertible B\"acklund 
transformation. 
Near $(U_0,V_0)$ with $B(U_0,V_0)=0$, we can solve $B(U,V)=0$ for $V$ by $V=V(U)$ with
differential $dV=\Pi(U,V(U))$, $\Pi = - B_V^{-1}B_U$.
Transforming $\Phi$ in the natural way gives an operator $\Psi$ defined near $V_0$ by
\begin{equation*}
	\Psi(V) =\Pi(U(V),V)\ \Phi(U(V))\ \big(\Pi(U(V),V)\big)^{-1} .
\end{equation*}
Fokas and Fuchssteiner proved in \cite{Fokas:Fuchssteiner:1981} that $\Psi$ is a strong symmetry 
of \eqref{vector field L}. 
Furthermore, if $\Phi$ in addition is hereditary, $\Psi$ is hereditary, too.
Finally, the correspondence generated by a B\"acklund transformation extends to hierarchies\footnote{
Actually, the arguments in \cite{Fokas:Fuchssteiner:1981} are given in the classical setting ($\mathcal{A}=\mathbb{C}$). 
However, it is easily seen that everything remains valid for evolution equations with noncommuting dependent variables.
}. 

\begin{rem}
a) In Section \ref{section recursion operator} and \ref{section hierarchy}, we make formal use of these results 
in a mild sense. Instead of requiring local invertibility, we compute under the assumption that inverses of operators 
like $D$ and $D+C_V$ exist.

b) Our short introduction to B\"acklund transformations is adapted to our needs. For a broader view on the topic 
we refer to the monographs \cite{Gu:Hu:Zhou:2005, Rogers:Schief:2002}.
\end{rem}

\section{Derivation of the mirror meta-mKdV recursion operator} 
\label{section recursion operator mirror}

In this appendix the proof of Theorem \ref{rec op} b) is provided, and an alternative derivation 
of the meta-mKdV recursion operator is given. We start with the following lemma.

\begin{lem} \label{more identities}
   \begin{itemize}
      \item[a)] For $T= \mp S_x S^{-1}$, the identity $(D\pm L_T)^{-1} R_S = R_S (D\pm C_T)^{-1}$ holds.
      \item[b)] For $T= \mp S^{-1} S_x$, the identity $(D\pm R_T)^{-1} L_S = L_S (D\mp C_T)^{-1}$ holds.
  \end{itemize}
\end{lem}
\begin{proof} 
   a) is Lemma \ref{identity}. b) Since $DL_S = L_SD+L_{S_x}$, 
   and since left- and right multiplication commute, $(D\pm R_T)L_S = L_S(D\pm R_T) + L_{S_x}$.
   Using $T= \mp S^{-1} S_x$, we can replace $L_{S_x}$ by $L_{\mp ST}= \mp L_SL_T$, and the lemma follows.
\end{proof}

\begin{proof}[Proof of Theorem \ref{rec op} b)]
Consider the Cole-Hopf link
\begin{eqnarray*}
   B(\tilde Q,V) &=& \tilde Q V + \tilde Q_x 
\end{eqnarray*}
between the mKdV \eqref{mkdv} and the mirror meta-mKdV \eqref{mirror meta mkdv}, see 
Proposition \ref{link from mirror meta}. Computing the Fr\'echet derivatives
\begin{eqnarray*}
	B_{\tilde Q} [\hat Q] =  \hat Q V + \hat Q_x , \qquad  B_V [\hat V] = \tilde Q \hat V ,
\end{eqnarray*}
we get $\tilde \Phi(\tilde Q) = \Pi \ \Psi(V) \Pi^{-1} $ with $\Pi = B_{\tilde Q}^{-1} B_V  
=  (D+R_V)^{-1} L_{\tilde Q} = L_{\tilde Q} (D-C_V)^{-1}$, according to Lemma \ref{more identities}. Hence, 
\begin{eqnarray*}
   \tilde \Phi(\tilde Q) 
      &=& L_{\tilde Q} (D-C_V)^{-1} \ (D - C_VD^{-1}C_V) (D - A_VD^{-1}A_V) \  (D-C_V) L_{\tilde Q^{-1}} \\
      &=& L_{\tilde Q} D^{-1} (D + C_V) (D - A_VD^{-1}A_V)  (D-C_V)  L_{\tilde Q^{-1}}  
\end{eqnarray*}
with $V=\tilde Q^{-1} \tilde Q_x$.
\end{proof}

As a supplement we give an alternative derivation of the meta-mKdV recursion operator 
starting with the B\"acklund transformation
\begin{eqnarray*}
     B(Q,\tilde V) &=& Q\tilde V - Q_x ,
\end{eqnarray*}
which links the meta-mKdV \eqref{meta mkdv} to the alternative mKdV \eqref{amkdv}, see Theorem \ref{link from meta}. 
In this case the Fr\'echet derivatives are
\begin{eqnarray*}
	B_Q [\hat Q] \ = \ \hat Q\tilde V - \hat Q_x , \qquad  B_{\tilde V} [\hat V]  \ = \ Q \hat V ,
\end{eqnarray*}
yielding the transformation operator $\Pi = B_{Q}^{-1} B_{\tilde V} = - (D-R_{\tilde V})^{-1} L_Q = - L_Q (D+C_{\tilde V})^{-1}$, 
by Lemma \ref{more identities}.
Hence a recursion operator for the meta-mKdV follows from
\begin{eqnarray*}
   \Phi(Q) &=&\Pi \ \tilde \Psi(\tilde V) \Pi^{-1} \\
      &=& L_Q \tilde D^{-1} \cdot
      		({\tilde D}+C_{\tilde V}) (\tilde D-A_{\tilde V}) \tilde D^{-1} (\tilde D+A_{\tilde V})(\tilde D-C_{\tilde V}) \tilde D^{-1} 
		\cdot  \tilde D L_{Q^{-1}} \\
      &=& L_Q 
      		\tilde D^{-1}  ({\tilde D}+C_{\tilde V}) (\tilde D-A_{\tilde V}\tilde D^{-1} A_{\tilde V})(\tilde D-C_{\tilde V})  
		L_{Q^{-1}} ,
\end{eqnarray*}
where $\tilde D = D+C_{\tilde V}$. Note that this is the representation of the meta-mKdV recursion operator given in 
Proposition \ref{rec op derivation}.

An analogous derivation of the mirror meta-mKdV recursion operator from the alternative amKdV \eqref{amkdv} using 
Proposition \ref{link from mirror meta} gives
\begin{equation*}
   \tilde \Phi (\tilde Q) = R_{\tilde Q} {\tilde D}^{-1} ({\tilde D}+C_{\tilde V})
								({\tilde D}-A_{\tilde V}{\tilde D}^{-1}A_{\tilde V}) 
								           ({\tilde D}-C_{\tilde V}) R_{\tilde Q^{-1}} ,
\end{equation*}
compare again Proposition \ref{rec op derivation}.

\end{appendix}

\quad

\end{document}